\documentclass[a4paper]{article}

\usepackage[english]{babel}
\usepackage[utf8x]{inputenc}
\usepackage[T1]{fontenc}
\usepackage{physics}

\usepackage[a4paper,top=3cm,bottom=2cm,left=3cm,right=3cm,marginparwidth=1.75cm]{geometry}

\usepackage{amsmath}
\usepackage{amssymb}
\usepackage{amsthm}
\usepackage{graphicx}
\usepackage{subcaption}
\usepackage{comment}
\usepackage[numbers]{natbib}
\usepackage[colorinlistoftodos]{todonotes}
\usepackage[colorlinks=true, allcolors=blue]{hyperref}
\usepackage{tikz-cd}

\newtheorem{theorem}{Theorem}[section]

\newtheorem{proposition}{Proposition}[section]
\newtheorem{lemma}{Lemma}[section]
\newtheorem{remark}{Remark}[section]
\newtheorem*{theorem*}{Theorem}
\newtheorem{definition}{Definition}[section]

\newtheorem{question}{Question}[section]

\usepackage[affil-it,noblocks]{authblk}

\title{Operator system characterizations of SIC-POVMs and mutually unbiased bases}
\author{Travis B. Russell}
\date{}

\begin{document}
\maketitle

\begin{abstract}
    We show that a symmetric informationally-complete positive operator-valued measure exists in a given dimension $d$ if and only if there exists a $d^2$-dimensional operator system satisfying certain order-theoretic conditions. We also describe a method of constructing such an operator system and demonstrate that the first step of this construction can be carried out successfully. We obtain analogous results for the existence of $d+1$ mutually unbiased bases in a given dimension.
\end{abstract}

\section{Introduction}

Let $d$ be a positive integer. A symmetric informationally-complete positive operator-valued measure (SIC-POVM) is a set of positive-semidefinite rank-one $d \times d$ complex matrices $F_1, F_2, \dots, F_{d^2}$ which span the $d \times d$ matrices, have constant trace, and whose Hilbert-Schmidt inner product is pairwise constant, i.e. there exist constants $C,D$, depending only on the dimension $d$, such that $\Tr(F_a)=D$ and $\Tr(F_a F_b) = C$ whenever $a \neq b$. Equivalently, the arguments of \cite{SICs2004} show that a SIC-POVM exists in dimension $d$ if and only if there exist unit vectors $\varphi_1, \varphi_2, \dots, \varphi_{d^2} \in \mathbb{C}^d$ which are \textit{equiangular}, meaning that there exists a constant $\lambda$, depending only on the dimension $d$, such that $|\langle \varphi_a, \varphi_b \rangle|^2 = \lambda$ whenever $a \neq b$. The equivalence between the two definitions can be obtained by viewing the operator $F_a$ as a scalar multiple of the rank-one projection onto the span of $\varphi_a$.

SIC-POVMs were first studied by Zauner in \cite{ZaunerThesis}. It is known that SIC-POVMs exist in every dimension up to $d=52$, and numerical evidence suggest that they likely exist in every dimension up to $d=200$, plus a few other larger dimensions. It is currently unknown if a SIC-POVM exists in every dimension $d$, or if there exists a SIC-POVM in infinitely many dimensions. For a summary of these results, see Section A of \cite{HorodeckiRZOpenProbsQIT}.

Suppose we are given equiangular vectors $\varphi_1, \varphi_2, \dots, \varphi_{d^2} \in \mathbb{C}^{d}$. For each $a=1,2,\dots,d^2$, let $P_a$ denote the projection onto the one-dimensional span of $\varphi_a$ in $\mathbb{C}^d$. Then the resulting set of projections $P_1, P_2, \dots, P_{d^2}$ satisfy the properties
\begin{equation} \label{eqn: Projection Conditions} I = \frac{1}{d} \sum_{a=1}^{d^2} P_a \quad \text{and} \quad P_a P_b P_a = \frac{1}{d+1} P_a  \end{equation}
where $I$ is the identity operator on the Hilbert space $\mathbb{C}^d$. Conversely, one could consider a situation in which we have projection operators $P_1, P_2, \dots, P_{d^2}$ satisfying the conditions (\ref{eqn: Projection Conditions}) but acting on an arbitrary (possibly infinite-dimensional) Hilbert space. In this paper, we will show that whenever this situation arises, a SIC-POVM exists, provided that an additional technical condition is satisfied. We also describe a method for constructing projections satisfying these conditions.

To explain the missing condition, we must introduce the tools used to prove the result. An \textit{operator system} is a unital self-adjoint vector subspace of $B(H)$, the set of bounded operators on a Hilbert space. Operator systems were abstractly characterized by Choi and Effros \cite{CHOIEffros1977}, who showed that there exists a correspondence between \textit{concrete} operator systems in $B(H)$ and certain \textit{matricially ordered} vector spaces which we call \textit{abstract operator systems}. A matricially ordered vector space consists of a complex vector space $V$, equipped with a conjugate linear involution $*: V \to V$, and a set of cones $C_n \subseteq M_n(V)$ which are compatible under the operation of conjugation by scalar matrices.

In recent work \cite{AR2020Published}, it has been shown that how the abstract data of an operator system may be used to identify projection operators. Using this result, the authors of \cite{AR2023Preprint} developed a method for constructing operator systems spanned by projections and satisfying linear relations, such as the first equation in (\ref{eqn: Projection Conditions}). In this paper, we show that the second condition in (\ref{eqn: Projection Conditions}) can also be encoded in the abstract structure of an operator system and that this condition can be incorporated into a method for constructing a set of projections satisfying both conditions.

The missing third condition needed for our main result is what we call \textit{$d$-minimality}. In broad terms, a $d$-minimal operator system $V$ is one whose entire matrix ordering is determined by the positive cone of $M_d(V)$. The abstract relations between the positive cones of $d$-minimal operator systems were recently characterized in \cite{ART2024Published}. They were also incorporated into the constructions developed in \cite{AR2023Preprint} for the purpose of studying the distinction between quantum and quantum-commuting correlations in the recently solved Tsirelson conjecture \cite{ji2020mip}. Our main result is the following:

\begin{theorem*}
    A SIC-POVM exists in dimension $d$ if and only if there exists a $d$-minimal operator system $V$ with unit $e$ spanned by projections $p_1, p_2, \dots, p_{d^2}$ satisfying
    \[ e = \frac{1}{d} \sum_{a=1}^{d^2} p_a \]
    and such that $p_a p_b p_a = \frac{1}{d+1} p_a$ for every $a \neq b$.
\end{theorem*}

Moreover, we describe how to construct such an operator system by an inductive process. The construction takes an initial operator system $V$ with its matrix ordering $\{C_n\}_{n=1}^{\infty}$, and inductively produces larger matrix orderings $\{C_n^{(k)}\}_{n=1}^{\infty}$ so that $C_n \subseteq C_n^{(1)} \subseteq C_n^{(2)} \subseteq \dots$. Letting $C^{\infty}_n$ denote the closure of $\cup_k C_n^{(k)}$, we find that the resulting operator system satisfies the conditions of the above theorem if and only if $C^{\infty}_1 \cap -C^{\infty}_1 = \{0\}$, i.e as long as the cones $C_n^{(k)}$ do not grow too large. We also give an explicit description for a family of initial matrix orderings $\{C_n\}$ parameterized by increasing sequences $\{t_n\}$ of real numbers. We show that a SIC-POVM exists if and only if the construction is successful for some sequence of real numbers $\{t_n\}$ and corresponding initial matrix ordering $\{C_n\}$.

The definition of a SIC-POVM is reminiscent of another notion studied in quantum information theory. A collection of orthonormal bases $\{\varphi_a^1\}_{a=1}^d, \{\varphi_a^2\}_{a=1}^d, \dots, \{\varphi_a^k\}_{a=1}^d$ for $\mathbb{C}^d$ are called \textit{mutually unbiased} if $|\langle \varphi_a^x, \varphi_b^y \rangle|^2 = \frac{1}{d}$ whenever $x \neq y$. In a given dimension $d$, it is known that there exist at most $d+1$ mutually unbiased bases \cite{WOOTTERS1989363}. Moreover, if $d = p^{\alpha}$ for some prime $p$ and positive integer $\alpha$, then it is known that there always exists $d+1$ mutually unbiased bases \cite{WOOTTERS1989363}. However, if $d$ has more than one prime factor, it is unknown whether or not there exist $d+1$ mutually unbiased bases \cite{HorodeckiRZOpenProbsQIT}. In particular, it unknown if there exist $7$ mutually unbiased bases in dimension $d=6$ \cite{RaynalLuEnglertMUBsSix2011}.

Suppose we have $d+1$ mutually unbiased bases $\{\varphi_a^1\}_{a=1}^d, \{\varphi_a^2\}_{a=1}^d, \dots, \{\varphi_a^{d+1}\}_{a=1}^d$ and let $P_a^x$ denote the rank-one projection onto the span of the vector $\varphi_a^x$. Then the resulting set of projections $\{P_a^x\}$ satisfy
\begin{equation} \label{eqn: MUB Conditions} I = \sum_{a=1}^d P_a^x \quad \text{and} \quad P_a^x P_b^y P_a^x = \frac{1}{d} P_a^x \end{equation}
for every $a,b = 1,2,\dots,d$ and $x \neq y$. Conversely, one could consider the situation in which one has projection operators acting on an arbitrary Hilbert space $H$ and satisfying the conditions (\ref{eqn: MUB Conditions}). Using similar techniques, we can show that these conditions, together with $d$-minimality, characterize the existence of $d+1$ mutually unbiased bases.

\begin{theorem*}
    There exist $d+1$ mutually unbiased bases in dimension $d$ if and only if there exists a $d$-minimal operator system $W$ with unit $e$ spanned by projections \[ \{p_a^x : a = 1,2,\dots,d \quad x=1,2,\dots,d+1 \} \]
    satisfying
    \[ e = \sum_{a=1}^d p_a^x \]
    for every $x=1,2,\dots,d+1$ and such that $p_a^x p_b^y p_a^x = \frac{1}{d} p_a^x$ whenever $x \neq y$.
\end{theorem*}

\noindent We also explain how such an operator system may be constructed and outline the details of the initial step of this construction.

Our paper is organized as follows. In Section \ref{sec: Prelims}, we cover preliminary results concerning operator systems. Most of these results are found in the literature, although we also introduce some new definitions and theorems that will be used later in the paper. In Section \ref{sec: Characterizations}, we prove the two characterization theorems above. Finally, in Section \ref{sec: constructions}, we present methods for constructing operator systems satisfying the conditions outlined in the characterization theorems.

\section{Preliminaries} \label{sec: Prelims}

Throughout this paper, we will assume familiarity with Hilbert spaces and their operators as well a fundamentals of C*-algebras. For completeness, we present most results needed concerning operator systems and completely positive maps. Readers seeking more details on this topic are referred to \cite{paulsen2002completely}. We let $M_n$ and $M_{n,k}$ denote the $n \times n$ and $n \times k$ complex matrices, respectively. We let $M_n^+$ denote the positive semidefinite matrices. Given a vector space $S$, we let $M_n(S)$ denote the vector space $M_n \otimes S$, usually regarded as the vector space of matrices with $S$-valued entries.

\subsection{Operator systems}

The main results of this paper are phrased in the language of \textit{operator systems}. These are ordered vector spaces which can be defined concretely as certain subspaces of $B(H)$ for a Hilbert space $H$ or abstractly as a matricially-ordered vector space satisfying certain compatibility conditions. While the two definitions are in some sense equivalent (see Theorem \ref{thm: Choi-Effros} below), there are some important subtleties which will make it necessary to distinguish between the concrete and abstract notions.

A \textit{concrete operator system} is a unital self-adjoint subspace of $B(H)$ for some Hilbert space $H$. Given a concrete operator system $S \subseteq B(H)$ and a positive integer $n$, we also regard $M_n(S) \subseteq B(H^n)$ as a concrete operator system with unit $I_k$, where $M_n(S)$ acts on $H^n$ via matrix-vector multiplication, i.e.
\[ (t_{ij}) (h_k) = (\sum_i t_{ik} h_k)_k. \]
If we let $C_n \subseteq M_n(S)$ denote the positive operators, it is easily verified that $C_n \oplus C_m \subseteq C_{n+m}$ and $\alpha^* C_n \alpha \subseteq C_k$ for every $\alpha \in M_{n,k}$. Letting $M_n(S)_h$ denote the self-adjoint elements of $M_n(S)$, it is easy to see that whenever $x \in M_n(S)$ and $x=x^*$, we have $x + \|x\|I_n \in C_n$ and that $x \in C_n$ if and only if $x + \epsilon I_n \in C_n$ for every $\epsilon > 0$.

We now define an abstract operator system. Let $S$ be a complex vector space. An \textit{involution} on $S$ is a map $*: S \to S$, denoted by $x \mapsto x^*$, which conjugate-linear and involutive, i.e. $(x+\lambda y)^* = x^* + \overline{\lambda}y^*$ and $(x^*)^* = x$ for every $x,y \in S$ and $\lambda \in \mathbb{C}$. A complex vector space $S$ equipped with an involution $*$ is called a $*$-vector space. The involution $*$ extends to $M_n(S)$ by setting the $(i,j)$ entry of $(a_{ij})^*$ equal to $a_{ji}^*$ for every matrix $(a_{ij}) \in M_n(S)$. We let $M_n(S)_h$ denote the elements $x \in M_n(S)$ such that $x=x^*$. A \textit{matrix-ordering} on $S$ is a sequence of subsets $C_n \subseteq M_n(S)_h$ such that $C_n \oplus C_m \subseteq C_{n+m}$ and $\alpha^* C_n \alpha \subseteq C_k$ for every $\alpha \in M_{n,k}$. A matrix ordering is \textit{proper} if $C_n \cap -C_n = \{0\}$ for every $n$ (equivalently, if $C_1 \cap -C_1 = \{0\}$). An element $e \in S$ is called a \textit{matrix order unit} if for every positive integer $n$ and every $x \in M_n(S)_h$ there exists $t > 0$ such that $x + t I_n \otimes e \in M_n(S)$. A matrix order unit is called \textit{Archimedean} if whenever $x + \epsilon I_n \otimes e \in C_n$ for every $\epsilon > 0$ it follows that $x \in C_n$. An \textit{abstract operator system} is a triple $(S,\{C_n\}_{n=1}^\infty, e)$ consisting of a $*$-vector space $S$, a proper matrix ordering $\{C_n\}_{n=1}^\infty$ for $S$, and an Archimedean matrix order unit $e \in S$. Each $C_n$ is a cone in $M_n(V)_h$ and its elements are called \textit{positive}. The cone $C_n$ induces a partial order on $M_n(V)_h$ given by $x \leq y$ whenever $y-x \in C_n$.

\begin{remark} \label{rmk: Archimedean closure}
    \emph{Suppose $S$ is a $*$-vector space with order unit $e \in S_h$ and that $\{D_n\}$ is a matrix ordering. If $e$ is not Archimedean, we may always replace $\{D_n\}$ by its \textit{Archimedean closure} defined as
    \[ C_n = \{x \in M_n(S)_h: x + \epsilon I_n \otimes e \in D_n \text{ for all } \epsilon > 0 \}. \]
    The resulting matrix ordering may fail to be proper (e.g. see \cite{PaulsenTodorovTomforde2011OSS}). However, if $\{C_n\}$ is proper, then $e$ will be an Archimdean matrix order unit and hence $(S, \{C_n\}, e)$ will be an operator system.}
\end{remark}

It is clear that every concrete operator system is an example of an abstract operator system, taking the involution to be the operator adjoint, the matrix ordering to be the positive operators, and the order unit to be the identity operator. In \cite{CHOIEffros1977}, Choi and Effros show that every abstract operator system may be identified with a concrete operator system. To explain what we mean by identifying operator systems, we introduce some more terminology. Given a linear map $\varphi: V \to W$ between vector spaces $V$ and $W$, we let $\varphi^{(n)}: M_n(V) \to M_n(W)$ denote the $n\textsuperscript{th}$ inflation map which applies $\varphi$ to each entry of a matrix. If $V$ and $W$ are operator systems, then we say $\varphi$ is $n$-positive if $\varphi^{(n)}$ maps positive elements to positive elements. We say $\varphi$ is \textit{completely positive} if it is $n$-positive for every $n$. If $\varphi$ is completely positive and injective, then it is a \textit{complete order embedding} if $\varphi^{-1}$ is completely positive on $\varphi(V)$, and a \textit{complete order isomorphism} if it is a surjective complete order embedding. With this language in hand, we may state the Choi-Effros Theorem.

\begin{theorem}[Choi-Effros] \label{thm: Choi-Effros}
    Let $(S,\{C_n\}_{n=1}^\infty, e)$ be an abstract operator system. Then there exists a Hilbert space $H$ and a unital complete order embedding $\pi: S \to B(H)$.
\end{theorem}

The Choi-Effros Theorem implies that abstract and concrete operator systems are ``the same'' up to complete order isomorphism. However, some properties of operators which will be necessary for our results are not preserved by complete order isomorphism. This is because two completely order isomorphic concrete operator systems can generate non-isomorphic C*-algebras. 

A \textit{C*-cover} for an abstract operator system $S$ is a C*-algebra $\mathcal{A}$ together with a unital complete order embedding $i: S \to \mathcal{A}$ such that $\mathcal{A} = C^*(i(S))$, i.e. $\mathcal{A}$ is generated by $S$ as a C*-algebra. In general, an abstract operator system has many non-isomorphic C*-covers. However, there is a unique ``smallest'' C*-cover called the \textit{C*-envelope}, denoted $C^*_e(S)$. The C*-envelope satisfies a universal property. Let $i: S \to C^*_e(S)$ denote the unital complete order embedding of $S$ into $C^*_e(S)$, and let $\mathcal{A}$ be another C*-cover with embedding $j: S \to \mathcal{A}$. Then there exists a unique $*$-homomorphism $\pi: \mathcal{A} \to C^*_e(S)$ satisfying $\pi(j(x))=i(x)$ for every $x \in S$. The existence and uniqueness of the C*-envelope was proven by Hamana in \cite{hamana1979injective}.

\subsection{Abstract projections}

For our results, we will be interested in the properties of projection operators as elements of operator systems. An operator $P \in B(H)$ is a \textit{projection} if $P=P^2=P^*$. Suppose that $S$ is an operator system. We say that an element $p \in S$ is an \textit{abstract projection} if there exist a Hilbert space $H$ and a unital complete order embedding $\pi: S \to B(H)$ such that $\pi(p)$ is a projection. Equivalently, there exists a C*-cover $(\mathcal{A},j)$ for $S$ such that $j(p)$ is a projection in $\mathcal{A}$. Since C*-covers for operator systems are not unique, it is possible \textit{a priori} that $p$ is not always represented as a projection in every C*-cover. In fact, if $p$ is not equal to $0$ or the identity, then there always exists a Hilbert space $K$ and a unital complete order embedding $\rho: S \to B(K)$ such that $\rho(p)$ is not a projection (see Example 6.2 of \cite{ART2024Published}). However, every abstract projection is represented as a projection in the C*-envelope. This is because if $j: S \to \mathcal{A}$ is a C*-cover with $j(p)$ a projection, then by the universal property of the C*-envelope there exists a $*$-homomorphism $\pi: \mathcal{A} \to C^*_e(S)$ such that $\pi(j(x))=i(x)$ for every $x \in S$. Hence $i(p) = \pi(j(p)) = \pi(j(p)^2) = \pi(j(p))^2 = i(p)^2$ and similarly $i(p)=i(p)^*$, so that $i(p)$ is a projection in $C^*_e(S)$.

The question of how to characterize abstract projections in operator systems was studied extensively in \cite{AR2020Published}. To motivate those results, we recall a lemma characterizing when the compression $PTP$ of an operator $T$ by a projection $P$ is positive using the order structure alone.

\begin{lemma}[See Lemma 4.5 of \cite{AR2020Published}] \label{lem: fundamental projection}
    Let $P \in B(H)$ be a projection and suppose that $T=T^* \in B(H)$. Then $PTP \geq 0$ if and only if for every $\epsilon > 0$ there exists $t > 0$ such that \[ T + \epsilon P + t P^{\perp} \geq 0 \]
    where $P^{\perp} := I - P$.
\end{lemma}

Suppose that $S$ is an operator system and $p \in S$ is an abstract projection. Let $\pi: S \to B(H)$ be a unital complete order embedding such that $P = \pi(p)$ is a projection. Suppose that $x=x^* \in M_n(S)$, and let $X = \pi^{(n)}(x)$. Then $X=X^*$. Since $P_n := I_n \otimes P$ is a projection in $B(H^n)$, we may write $H^n = P_n H^n \oplus P_n^{\perp} H^n$. With respect to this decomposition of $H^n$, we may regard $X$ as a $2 \times 2$ operator matrix. From this perspective, it is easy to see that $X$ is positive if and only if the compression of
\[ \begin{pmatrix} X & X \\ X & X \end{pmatrix} \in B(H^{2n}) \]
by the projection $P_n \oplus P_n^{\perp}$ is positive. By Lemma \ref{lem: fundamental projection} and the fact that $\pi^{-1}$ is a complete order isomorphism on the range of $\pi$, we see that $x$ is positive if and only if for every $\epsilon > 0$ there exists $t > 0$ such that
\[ \begin{pmatrix} x & x \\ x & x \end{pmatrix} + \epsilon P_n \oplus P_n^{\perp} + t P_n^{\perp} \oplus P_n \in C_{2n} \]
where $C_{2n}$ denotes the positive cone in $M_{2n}(S)$. It turns out that this property characterizes abstract projections in operator systems.

In the following theorem, we say that $p x p \geq 0$ \textbf{abstractly} if for every $\epsilon > 0$ there exist $t > 0$ such that $x + \epsilon p + t p^{\perp} \geq 0$, where $p^{\perp} := e - p$ and $e$ is the unit of the operator system.

\begin{theorem}[See Theorem 5.10 of \cite{AR2020Published} and Lemma 3.6 of \cite{ARTaPublished2022}] \label{thm: projection characterizations}
    Let $(S,\{C_n\},e)$ be an abstract operator system and suppose $p \in  S$ satisfies $0 \leq p \leq e$. Set $p^{\perp} = e - p$. Let $C_n(p)$ denote the set of elements $x \in M_n(S)_h$ with the property that 
    \[ \left[ (p \oplus p^{\perp})  \otimes I_n \right] \begin{pmatrix} x & x \\ x & x \end{pmatrix} \left[ (p \oplus p^{\perp}) \otimes I_n \right] \geq 0  \]
    abstractly. If $\{C_n(p)\}$ is a proper matrix ordering, then $p$ is abstract projection in the operator system $(S,\{C_n(p)\},e)$. Moreover, $p$ is an abstract projection in $S$ if and only if for every $n \in \mathbb{N}$, $C_n = C_n(p)$.
\end{theorem}

Theorem \ref{thm: projection characterizations} abstractly characterizes projections in operator systems, meaning that any element $p \in S$ which satisfies the conditions of the Theorem must be represented as a projection in some C*-cover for $S$. In particular, every element satisfying these conditions is represented as a projection in $C^*_e(S)$. We will make use of this property frequently.

\subsection{Abstract relations} \label{subsec: Abstract relations}

We now introduce the notion of \textit{abstract relations}, which will be crucial to characterizing SIC-POVMs and mutually unbiased bases in the language of operator systems.

\begin{definition}
    Let $S$ be an operator system and suppose that $p \in S$ is a positive contraction. Then we say $pxp = 0$ \textbf{abstractly} if $pxp \geq 0$ abstractly and $p(-x)p \geq 0$ abstractly.
\end{definition}

\noindent We emphasize that the above definition does not require $p$ to be an abstract projection. In Section \ref{sec: constructions}, we will produce examples of operator systems where $pxp=0$ abstractly although $p$ is not an abstract projection. However, the following observation explains the importance of this notion in the case when $p$ is an abstract projection.

\begin{proposition}
    Suppose that $S$ is an operator system, $i: S \to C^*_e(S)$ is the embedding into its C*-envelope, $p \in S$ is an abstract projection, and $pxp = 0$ abstractly. Then $i(p)i(x)i(p) = 0$ in $C^*_e(S)$.
\end{proposition}

\begin{proof}
    By Theorem \ref{thm: projection characterizations}, $i(p)$ is a projection. Since $pxp = 0$ abstractly, $\pm i(p)i(x)i(p) \geq 0$ in $C^*_e(S)$ by Lemma \ref{lem: fundamental projection}. Since the positive cone of $C^*_e(S)$ is proper, $i(p)i(x)i(p) = 0$ in $C^*_e(S)$.
\end{proof}

With the above notions, we can define a variety of \textbf{abstract relations}. For example, suppose we want the relation $pxp = t p$ to hold in $C^*_e(S)$, where $p,x \in S$, $t \in \mathbb{R}$, and $p$ is an abstract projection. This relation can be enforced by asking that $p(x-tI)p = 0$ abstractly. 

\subsection{$d$-minimality}

Given an abstract operator system $S$, the Choi-Effros Theorem guarantees the existence of a Hilbert space $H$ and a complete order embedding $\pi: S \to B(H)$. However, even if $S$ is finite dimensional, there is no guarantee that $H$ is finite-dimensional. In fact, there are many examples of finite dimensional operator systems which require an infinite dimensional Hilbert space $H$ to produce a complete order embedding $\pi: S \to B(H)$ (c.f. \cite{PaulsenGroups}). Since we are interested in finite-dimensional phenomena, namely SIC-POVMs and mutually unbiased bases, it would be helpful to somehow bound the dimension of a C*-cover for a given operator system. This can be partially accomplished by considering \textit{$d$-minimal} operator systems. This notion was introduced in \cite{xhabli2012super} and studied in the context of abstract projections in \cite{ART2024Published}.

Given a positive integer $d$, an operator system $(S, \{C_n\}, e)$ is \textbf{$d$-minimal} if $\varphi^{(n)}(x) \geq 0$ for all ucp $\varphi: S \to M_d$ implies that $x \in C_n$. Equivalently, the direct sum of all ucp maps $\varphi: S \to M_d$ defines a unital complete order embedding of $S$ into an infinite direct sum of $d \times d$ matrix algebras. In other words, $S$ can be regarded as a concrete operator subsystem of the $d \times d$ block diagonal operators on a Hilbert space with respect to some orthonormal basis. The following theorem from \cite{ART2024Published} gives a more intrinsic characterization of $d$-minimal operator systems.

\begin{theorem}[Section 3 of \cite{ART2024Published}] \label{thm: d-minimal characterization}
    Let $(S,\{C_n\},e)$ be an abstract operator system. Then $S$ is $d$-minimal if and only if for every $n \in \mathbb{N}$, $C_n$ is equal to the set of $x \in M_n(S)_h$ such that $\alpha^* x \alpha \in C_d$ for every $\alpha \in M_{n,d}$.
\end{theorem}

For any matrix ordering $\{C_n\}$, $x \in C_n$ implies that $\alpha^* x \alpha \in C_d$ for every $\alpha \in M_{n,d}$. Hence, the content of Theorem \ref{thm: d-minimal characterization} is the converse, so that the entire matrix ordering is uniquely determined by the cone $C_d$. Now let $\{C_n\}$ be any matrix ordering making $(S,\{C_n\},e)$ into an operator system. Then we can replace $\{C_n\}$ with another matrix ordering $\{(C_n)^{d-\text{min}} \}$ making $(S,\{(C_n)^{d-\text{min}} \},e)$ into a $d$-minimal matrix ordering. The following explains how $\{(C_n)^{d-\text{min}} \}$ is defined and its relation to the original matrix ordering $\{C_n\}$.

\begin{proposition} \label{Prop: d-min matrix ordering}
    Let $\{C_n\}$ be any matrix ordering making $(S,\{C_n\},e)$ into an operator system. Then for $x \in M_n(S)_h$ the following statements are equivalent.
    \begin{enumerate}
        \item For every ucp map $\varphi: S \to M_d$, $\varphi^{(d)}(x) \geq 0$.
        \item For every $\alpha \in M_{n,d}$, $\alpha^* x \alpha \in C_d$.
    \end{enumerate}
    Let $(C_n)^{d-\text{min}}$ denote the set of all $x \in M_n(S)_h$ satisfying these equivalent conditions. Then $\{(C_n)^{d-\text{min}} \}$ is a matrix-ordering making $(S,\{(C_n)^{d-\text{min}} \},e)$ into a $d$-minimal operator system. Moreover, $C_n \subseteq (C_n)^{d-\text{min}}$ for every $n \in \mathbb{N}$ and $C_n = (C_n)^{d-\text{min}}$ for every $n \leq d$.
\end{proposition}

We conclude this section by considering how projections behave in $d$-minimal operator systems. Suppose that $S$ is $d$-minimal and that $p \in S$ is an abstract projection. Then there exists a complete order embedding $\pi: S \to B(H)$ such that $\pi(p)$ is a projection and a (possibly different) complete order embedding $\pi': S \to B(K)$ such that $\pi'(S)$ has a block diagonal form with respect to some orthonormal basis. If these representations coincide, then we may conclude that $\pi(p)$ is a direct sum of projections on matrix algebras. Since an operator system may have many C*-covers, it is not immediately clear if this situation occurs. However, we have the following result from \cite{ART2024Published}.

\begin{theorem} \label{thm: C*-envelope is d-minimal}
    Suppose that $(S,\{C_n\}, e)$ is $d$-minimal. If $\pi: C^*_e(S) \to B(H)$ is an irreducible representation of $C^*_e(S)$, then $\dim(H) \leq d$. Hence there exists a faithful $*$-representation of $C^*_e(S)$ onto a direct sum of matrix algebras each with dimension no larger than $d$, namely the direct sum of all irreducible representations of $C^*_e(S)$ (c.f. Corollary I.9.11 of \cite{DavidsonCstarBook}).
\end{theorem}

Since every abstract projection is a projection in $C^*_e(S)$, we conclude that abstract projections in $d$-minimal systems may be regarded as direct sums of projections in matrix algebras of size no larger than $d \times d$.

\subsection{Inductive limits} \label{subsec: Inductive limits}

In Section 4 below, we will consider a method of constructing operator systems which contain abstract projections, satisfy abstract relations, and are $d$-minimal. These constructions are achieved by an inductive limit process. Specifically, we we begin with an operator system $(V, \{C_n^{(0)}\}, e)$ and then inductively defined an increasing sequence $\{C_n^{(k)}\}$ of matrix orderings so that $(V, \{C_n^{(k)}\}, e)$ is an operator system and the identity map $i: (V, \{C_n^{(k)}\}, e) \to (V, \{C_n^{(k+1)}\}, e)$ is completely positive. In the limit, we obtain a matrix ordering $\{C_n^{(\infty)}\}$ defined below.

\begin{definition}[Inductive limit of matrix orderings]
    Let $V$ be a $*$-vector space, together with an element $e \in V$ such that $e=e^*$, and suppose that for each $k \in \mathbb{N}$ we have a matrix ordering $\{C_n\}_{n=1}^\infty$ such that
    \begin{enumerate}
        \item $(V, \{C_n^{(k)}\}_{n=1}^\infty, e)$ is an operator system,
        \item $C_n^{(k)} \subseteq C_n^{(k+1)}$ for every $n,k \in \mathbb{N}$.
    \end{enumerate}
    For each $n \in \mathbb{N}$, define $C^{(\infty)}_n$ to be the set of elements $x \in V$ such that $x=x^*$ and for every $\epsilon > 0$ there exists $n \in \mathbb{N}$ such that $x + \epsilon I_n \otimes e \in C_n^{(k)}$. In this situation, the sequence $\{C_n^{(k)}\}_{n=1}^\infty, k=1,2,\dots$ is called an \textbf{inductive sequence of matrix orderings} and the family of sets $\{C_n^{(\infty)}\}$ is called the \textbf{inductive limit} of the sequence of matrix orderings.
\end{definition}

It was shown \cite{ARTaPublished2022} that $\{C_n^{(\infty)}\}$ is a matrix ordering and is Archimedean closed, meaning that if $x + \epsilon I_n \otimes e \in C_n^{(\infty)}$ for every $\epsilon > 0$ then $x \in C_n^{(\infty)}$. However, it may be the case that $\{C^{(\infty)}_n\}$ is not proper, meaning that both $x$ and $-x$ are elements of $C_n^{(\infty)}$ for some $x \neq 0$.

The reason for involving inductive limits of matrix orderings is to take advantage of the following two results, which guarantee that an inductive limits of matrix orderings respect abstract projections and $d$-minimality.

\begin{theorem}[See Proposition 4.3 of \cite{AR2023Preprint}] \label{thm: inductive limits respect abstract projections}
    Let $V$ be a $*$-vector space with self-adjoint elements $p, e \in V$. Suppose that $\{C_n^{(k)}\}_{n=1}^\infty, k=1,2,\dots$, is an inductive sequence of matrix orderings such that for every $k \in \mathbb{N}$, $(V, \{C_n^{(k)}\}_{n=1}^\infty, e)$ is an operator system with abstract projection $p$. Then $p$ is an abstract projection in $(V, \{C_n^{(\infty)}\}_{n=1}^\infty, e)$ provided that $\{C^{(\infty)}_n\}_{n=1}^\infty$ is proper.
\end{theorem}

The next result is new as stated, but builds on the arguments of \cite{AR2023Preprint} (see, e.g. Lemma 5.8 of \cite{AR2023Preprint}). Before giving the proof, we recall a few elementary facts. First, if $V$ is an operator system and $\varphi: V \to M_d$ is a unital completely positive map, then $\varphi$ is completely contractive (c.f. Proposition 3.6 of \cite{paulsen2002completely}). Second, the set of all unital completely contractive maps from an operator system $V$ to $M_d$ is weak-$*$ compact (c.f Theorem 7.4 of \cite{paulsen2002completely}).

\begin{theorem} \label{thm: inductive limits respect k-minimality}
    Let $V$ be a $*$-vector space with self-adjoint element $e \in V$ and let $d \in \mathbb{N}$. Suppose that $\{C_n^{(k)}\}_{n=1}^\infty, k=1,2,\dots$, is an inductive sequence of matrix orderings such that for every $k \in \mathbb{N}$, $(V, \{C_n^{(k)}\}_{n=1}^\infty, e)$ is a $d$-minimal operator system. Then $(V, \{C_n^{(\infty)}\}_{n=1}^\infty, e)$ is a $d$-minimal operator system provided that $\{C^{(\infty)}_n\}_{n=1}^\infty$ is proper.
\end{theorem}

\begin{proof}
    Let $x \in M_n(V)$ such that $x=x^*$ and suppose that for every unital linear map $\varphi: V \to M_d$ satisfying $\varphi^{(d)}(C_d^{(\infty)}) \subseteq M_{d^2}^+$ we have $\varphi^{(n)}(x) \geq 0$. We claim that $x \in C_n^{(\infty)}$. If this holds, then since $x$ was chosen arbitrarily we may conclude that $(V, \{C_n^{(\infty)}\}, e)$ is $d$-minimal by Proposition \ref{Prop: d-min matrix ordering}.

    To reach a contradiction, suppose that $x \notin C_n^{(\infty)}$. Then there exists $\epsilon > 0$ such that $x + \epsilon I_n \otimes e \notin C_n^{(k)}$ for all $k \in \mathbb{N}$. Because $(V, \{C_n^{(k)}\}, e)$ is $d$-minimal, there exists a unital linear map $\varphi_k: V \to M_d$ such that $\varphi_k^{(d)}(C_d^{(k)}) \subseteq M_{d^2}^+$ but $\varphi_k^{(n)}(x) + \epsilon I_{nd} = \varphi_k^{(n)}(x + \epsilon I_n \otimes e) \notin M_{nd}^+$. Let $\{\varphi_k: V \to M_d\}$ be a sequence of unital linear maps with these properties. Since $C_m^{(1)} \subseteq C_m^{(k)}$ for every $k,m \in \mathbb{N}$ and since $\varphi_k$ is completely positive with respect to the matrix ordering $\{C_m^{(k)}\}$, we conclude that each $\varphi_k$ is completely positive on the operator system $(V, \{C_m^{(1)}\}, e)$. Since the set of unital completely positive maps from this operator system to $M_d$ is weak-$*$ compact, there exists a unital linear map $\varphi: V \to M_d$ which is completely positive on the operator system $(V, \{C_m^{(1)}\}, e)$ and which is a weak-$*$ limit point of the sequence $\{\varphi_k: V \to M_d\}$.

    We claim that $\varphi: V \to M_d$ satisfies $\varphi^{(d)}(C_d^{(\infty)}) \subseteq M_{d^2}^+$. To see this, let $y \in C_d^{(\infty)}$ and let $\delta > 0$. Then there exists $N \in \mathbb{N}$ such that $y + \delta I_d \otimes e \in C_d^{(N)}$. Since $\varphi$ is the weak-$*$ limit of the sequence $\varphi_k$, the sequence $\varphi_k^{(d)}(y)$ converges in $M_d$. Moreover, since the sequence $\{C_d^{(k)}\}$ is nested and since $\varphi_k^{(d)}(C_d^{(k)}) \subseteq M_{d^2}^+$, we have $\varphi_k^{(d)}(y) + \delta I_{d^2} = \varphi_k^{(d)}(x + \delta I_d \otimes e) \in M_{d^2}^+$ for every $k > N$. We conclude that $\varphi^{(d)}(y) + \delta I_{d^2} \in M_{d^2}^+$ for every $\delta > 0$ and hence $\varphi^{(d)}(y) \geq 0$. However, this implies that $\varphi^{(n)}(x) \geq 0$. This is a contradiction since $\varphi_k^{(n)}(x) \to \varphi^{(n)}(x)$ and $\varphi_k^{(n)}(x) + \epsilon I_n \notin M_{nd}^+$ for every $k \in \mathbb{N}$.
\end{proof}

\section{Characterizations} \label{sec: Characterizations}

In this Section, we will prove two characterization theorems relating operator systems to SIC-POVMs and mutually unbiased bases. We begin by reformulating the existence of SIC-POVMs and maximal families of mutually unbiased bases in terms of rank one projections. These reformulations are very well-known and elementary, but we include a brief proof for completeness.

\begin{theorem} \label{Thm: SICs from projections}
Let $d$ be a positive integer, and let $\lambda = \frac{1}{d+1}$ and $\mu = \frac{1}{d}$.
\begin{enumerate}
    \item There exists a SIC-POVM in dimension $d$ if and only if there exist rank-one projections $P_1, P_2, \dots, P_{d^2} \in M_d$ satisfying
        \[ \sum_{i=1}^{d^2} P_i = dI \quad \text{and} \quad P_iP_jP_i = \lambda P_i \text{ for every } i \neq j. \]
    \item There exist $d+1$ mutually unbiased bases in dimension $d$ if and only if there exist projection-valued measures \[ \{P_a^1\}_{a=1}^d, \{P_a^2\}_{a=1}^d, \dots, \{P_a^{d+1}\}_{a=1}^d \in M_d \] satisfying $P_a^xP_b^yP_a^x = \mu P_a^x$ for every $x \neq y$ and any $a,b \in \{1,2,\dots,d\}$.
\end{enumerate}
\end{theorem}

\begin{proof}
    Statement 1 is proven in Section II of \cite{SICs2004} using a technique of Benedetto and Fickus \cite{BenedettoFickusFrames2003}. The value $\lambda = \frac{1}{d+1}$ follows from the equations
    \[ (d^2-1)\lambda + 1 = \sum_{a=1}^d \Tr(P_a P_b) = \Tr(d P_b) = d. \]
    
    For statement 2, suppose $\{ \varphi_a^1\}_{a=1}^d, \{\varphi_a^2\}_{a=1}^d, \dots, \{ \varphi_a^{d+1}\}_{a=1}^d$ are mutually unbiased bases for $\mathbb{C}^d$. Let $P_a^x$ be the rank one projection onto the span of $\varphi_a^x$ for each $a \in \{1,2,\dots,d\}$ and $x \in \{1,2,\dots,d+1\}$. Since $\{\varphi_a^x\}_{a=1}^d$ is an orthonormal basis, $\{P_a^x\}_{a=1}^d$ is a projection valued measure. Since $|\langle \varphi_a^x, \varphi_b^y \rangle |^2 = \mu$ for every $x \neq y$, the relation $P_a^xP_b^yP_a^x = \mu P_a^x$ holds for every $x \neq y$ and every $a,b \in \{1,2,\dots,d\}$. The value $\mu = \frac{1}{d}$ can be checked using the equations
    \[ d \mu = \sum_{a=1}^d \Tr(P_a^x P_b^y) = \Tr(P_b^y) = 1. \]

    Conversely, suppose we are given projection-valued measures \[ \{P_a^1\}_{a=1}^d, \{P_a^2\}_{a=1}^d, \dots, \{P_a^{d+1}\}_{a=1}^d \in M_d \] satisfying $P_a^xP_b^yP_a^x = \mu P_a^x$ for every $x \neq y$ and any $a,b \in \{1,2,\dots,d\}$. Let $x,y \in \{1,2,\dots, d+1\}$ with $x \neq y$. Since $P_a^x \neq 0$ for some $a \in \{1,2,\dots,d\}$, we must have $\Tr(P_a^x) \neq 0$. Since $\mu \Tr(P_a^x) = \Tr(P_a^x P_b^y P_a^x)$, it must be that $P_b^y \neq 0$. This holds for every $b \in \{1,2,\dots, d\}$, so each $P_b^y$ is a non-zero projection. Since $\sum_b P_b^y = I$, each $P_b^y$ is rank-one. Since $x$ and $y$ were arbitrary, this holds for every $y$. Thus we may choose unit vectors $\{\varphi_a^x\}$ in the range of $P_a^x$ for every choice of $a$ and $x$. Since $\{P_a^x\}_{a=1}^d$ is a PVM, $\{\varphi_a^x\}_{a=1}^d$ is an orthonormal basis for $\mathbb{C}^d$. Since $P_a^xP_b^yP_a^x = |\langle \varphi_a^x, \varphi_b^y \rangle|^2 P_a^x$ for every $x \neq y$ and every $a,b$, we have $|\langle \varphi_a^x, \varphi_b^y \rangle| = \mu$. Thus the bases $\{ \varphi_a^1\}_{a=1}^d, \{\varphi_a^2\}_{a=1}^d, \dots, \{ \varphi_a^{d+1}\}_{a=1}^d$ are mutually unbiased.
\end{proof}

In the following, we consider the situation where we have projections satisfying relations like those in the previous theorem. However, we do not assume the projections necessarily reside in the algebra of $d \times d$ matrices.

\begin{lemma} \label{lem: trace properties}
Let $d$ be a positive integer, and let $\lambda = \frac{1}{d+1}$ and $\mu = \frac{1}{d}$.
\begin{enumerate}
    \item Suppose $\mathcal{A}$ is a C*-algebra generated by projections $p_1, p_2, \dots, p_{d^2}$ satisfying
        \[ \sum_{i=1}^{d^2} p_i = dI \quad \text{and} \quad p_ip_jp_i = \lambda p_i \text{ for every } i \neq j. \]
    Then for any tracial state $\tau: \mathcal{A} \to \mathbb{C}$ we have
    \[ \tau(p_i) = \frac{1}{d} \quad \text{and} \quad \tau(p_ip_j) = \frac{\lambda}{d} \]
    for any $i, j \in \{1,2,\dots,d^2\}$ with $i \neq j$.

    \item Suppose there exists a C*-algebra $\mathcal{A}$ generated by projections \[ \{p_a^x: a=1,\dots,d; x=1,\dots,d+1\} \] satisfying
        \[ \sum_{a=1}^{d} p_a^z = I \quad \text{and} \quad p_a^xp_b^yp_a^x = \mu p_a^x \]
    for every $a,b \in \{1,2,\dots,d\}$ and every $x,y,z \in \{1, 2, \dots, d+1\}$ such that $x \neq y$. Then for any tracial state $\tau: \mathcal{A} \to \mathbb{C}$, we have
    \[ \tau(p_a^x) = \frac{1}{d} \quad \text{and} \quad \tau(p_a^x p_b^y p_a^x) = \frac{\mu}{d} = \frac{1}{d^2} \]
    for any $a,b \in \{1,2,\dots, d\}$ and $x,y \in \{1,2, \dots, d+1\}$ with $x \neq y$.
\end{enumerate}
    
\end{lemma}

\begin{proof}
    We prove the first statement, leaving the similar proof of the second statement to the reader. Let $\mathcal{A}$ be a C*-algebra as in the statement and let $\tau: \mathcal{A} \to \mathbb{C}$ be a tracial state. Suppose that $i,j \in \{1,2,\dots,d^2\}$ and that $i \neq j$. Let $u = \lambda^{-1/2} p_i p_j$. Then
    \[ uu^* = \lambda^{-1} p_i p_j p_j p_i = \lambda^{-1} p_i p_j p_i = \lambda^{-1} \lambda p_i = p_i \]
    and similarly $u^*u = p_j$. It follows that $u$ is a partial isometry with range projection $p_i$ and source projection $p_j$. Since $\tau$ is tracial, $\tau(p_i) = \tau(uu^*) = \tau(u^*u) = \tau(p_j)$. Hence the value of $\tau(p_a)$ is the same for all $a \in \{ 1,2,\dots,d^2\}$. Since $\tau$ is unital,
    \[ 1 = \tau(I) = \frac{1}{d} \sum_{i=1}^{d^2} \tau(p_i) = d \tau(p_i) \]
    for every $i \in \{1,2,\dots,d^2\}$. Hence $\tau(p_i) = \frac{1}{d}$ for each $i$. Furthermore, for $i \neq j$,
    \[ \tau(p_i p_j) = \tau(p_i^2 p_j) = \tau(p_i p_j p_i) = \lambda \tau(p_i) = \frac{\lambda}{d}. \]
    So $\tau(p_i p_j) = \frac{\lambda}{d}$ for all $i \neq j$.
\end{proof}

\begin{remark}
    \emph{Before moving on to operator systems and away from C*-algebras, we should mention some related C*-algebraic results found in the literature. Conditions 1-4 of \cite{Navascues2012Handbook} correspond to the conditions in part 2 of Lemma \ref{lem: trace properties}. They add an additional Condition 5 specifying relations of higher order products in the algebra generated by the projections $\{P^x_a\}$. If these 5 conditions are satisfied in a C*-algebra, they show that $d+1$ MUBs exist in dimension $d$ by means of semidefinite programming. Similar results, via wreath products, were recently found by Griblings and Polak in \cite{Gribling2024mutuallyunbiased}. For SIC-POVMs, the very recent preprint \cite{FarkasEtAl2024Preprint} considers, in Appendix D, the universal C*-algebra generated by projections satisfying the conditions in Part 1 Lemma \ref{lem: trace properties}. They show that if this Algebra has a representation on a $d$-dimensional Hilbert space, then a SIC-POVM exists in dimension $d$. We will recover this result as a consequence Theorem \ref{Theorem: SIC Characterization} below. However, we emphasize that our results only require operator systems satisfying order-theoretic conditions \textit{a priori}. We thank Daniel McNulty and anonymous referees for pointing out these references to us.}
\end{remark}

We can now provide an operator system characterization for the existence of SIC-POVM.

\begin{theorem} \label{Theorem: SIC Characterization}
    Let $d$ be a positive integer and let $\lambda = \frac{1}{d+1}$. Then there exists a SIC-POVM in dimension $d$ if and only if there exists an operator system $V$ with unit $e$ spanned by elements $p_1, p_2, \dots, p_{d^2}$ such that
        \begin{enumerate}
            \item $V$ is $d$-minimal,
            \item $\sum_{i=1}^{d^2} p_i = de$,
            \item each $p_i$ is an abstract projection, and
            \item the relation $p_i (p_j - \lambda e) p_i = 0$ holds abstractly for every $i \neq j$.
        \end{enumerate}
\end{theorem}

\begin{proof}
    First, suppose that there exists a SIC-POVM in dimension $d$. By Theorem \ref{Thm: SICs from projections}, there exist projections $P_1, \dots, P_{d^2} \in M_d$ satisfying $P_i P_j P_i = \lambda P_i$ for every $i \neq j$. Since $M_d$ is $d$-minimal and since the projections $P_1, \dots, P_{d^2}$ span $M_d$, we see that $V=M_d$ satisfies the required properties when viewed as an operator system. It remains to consider the other direction of the statement.

    Suppose $V$ is as stated in the Theorem. Let $\mathcal{A} = C^*_e(V)$. Then $\mathcal{A}$ is $d$-minimal by Theorem \ref{thm: C*-envelope is d-minimal}. Let $\pi: \mathcal{A} \to M_k$ be an irreducible representation for $\mathcal{A}$, so that $\pi(\mathcal{A}) = M_k$. Then $k \leq d$. We will show that $k = d$ and that the projections $P_i := \pi(p_i)$ satisfy the conditions of statement 2 in Theorem \ref{Thm: SICs from projections}.

    Since each $p_i$ is an abstract projection and since $p_i (p_j - \lambda e) p_i = 0$ abstractly for $i \neq j$, we have $P_i=P_i^*=P_i^2$ and $P_i P_j P_i = \lambda P_i$ for every $i \neq j$ in $\mathcal{A}$. Hence these relations hold in $\pi(\mathcal{A}) = M_k$ as well. Let $\tau$ denote the unique tracial state on $\pi(\mathcal{A})$, i.e. $\tau(x) = \frac{1}{k} \Tr(x)$ for every $x \in M_k$. By Lemma \ref{lem: trace properties}, we must have $\tau(p_i) = \frac{1}{d}$ and $\tau(p_i p_j) = \frac{\lambda}{d}$ for every $i\neq j$. Since the rank of $p_i$ is $k \tau(p_i)$, it must be that $k/d$ is an integer. Since $k \leq d$, we conclude $k=d$ and each $p_i$ is rank one.
\end{proof}

A similar statement holds for mutually unbiased bases. We omit the proof, since it is similar to the proof the preceding theorem.

\begin{theorem} \label{Theorem: MUB Characterization}
    Let $d$ be a positive integer and let $\mu = \frac{1}{d}$. Then there exist $d+1$ MUBs in dimension $d$ if and only if there exists an operator system $W$ with unit $e$ spanned by elements \[ \{p_a^x: a=1,\dots,d; x=1,\dots,d+1\} \] such that
        \begin{enumerate}
            \item $W$ is $d$-minimal,
            \item $\sum_{a=1}^d p_a^x = e$ for every $x$,
            \item each $p_a^x$ is an abstract projection, and
            \item the relation $p_a^x (p_b^y - \mu e) p_a^x = 0$ holds abstractly for every $x \neq y$.
        \end{enumerate}
\end{theorem}

\begin{remark} \emph{In the proof of Theorem \ref{Theorem: SIC Characterization}, we only use $d$-minimality to show that the range of an irreducible representation of the C*-envelope is $M_d$. It could be the case that in any C*-algebra generated by projections satisfying the relations $\sum p_i = dI$ and $p_i p_j p_i = \lambda p_i$ the irreducible representations have range $M_d$. If that is the case, then the assumption of $d$-minimality can be dropped. However, we do not know whether or not this is the case. The analogous remark can be made in the context of Theorem \ref{Theorem: MUB Characterization}. Therefore we pose the following questions:} \end{remark}

\begin{question}
    Let $\mathcal{A}$ be a C*-algebra. Suppose $p_1, p_2, \dots, p_{d^2} \in \mathcal{A}$ are projections satisfying $\sum p_i = dI$ and $p_i p_j p_i = \frac{1}{d+1} p_i$ for all $i \neq j$. If $\pi: \mathcal{A} \to B(H)$ is irreducible, must $\dim(H) \leq d$?
\end{question}

\begin{question}
    Let $\mathcal{A}$ be a C*-algebra. Suppose $\{p_a^x: a=1,\dots,d\} \subseteq \mathcal{A}$ is a projection-valued measure for each $x=1,2,\dots,d+1$, and that $p_a^x p_b^y p_a^x = \frac{1}{d} p_a^x$ for all $x \neq y$. If $\pi: \mathcal{A} \to B(H)$ is irreducible, must $\dim(H) \leq d$?
\end{question}

\section{Constructions} \label{sec: constructions}

In the previous section, we showed that the existence of a SIC-POVM (or $d+1$ mutually unbiased bases) in dimension $d$ is equivalent to the existence of an operator system satisfying certain properties. In this section, we demonstrate a method for constructing operator systems with the desired properties. This method will always produce a matrix ordered vector space. However, the matrix ordering may not be proper -- i.e. its cones may include members which are both positive and negative. The construction ``succeeds'' if it produces an operator system, i.e. if the matrix ordering is proper. Otherwise, the construction ``fails.'' We will show that the failure or success of this construction with some initial data is equivalent to the existence of a SIC-POVM or $d+1$ mutually unbiased bases in dimension $d$. We will also show that the first step of the construction can be carried out successfully for SIC-POVMs, leaving the similar construction for mutually unbiased bases to the interested reader.

Our constructions are based on those developed in \cite{ARTaPublished2022} and \cite{AR2023Preprint}. The initial data of the construction is an abstract operator system $(V,\{C_n\},e)$ and a finite set of positive contractions $\{p_1,p_2,\dots,p_n\} \subseteq V$. The construction yields an inductive sequence of matrix orderings $$\{C_n^{(1)}\}, \{C_n^{(2)}\}, \{C_n^{(3)}\}, \dots $$ (see Subsection \ref{subsec: Inductive limits}). The inductive limit matrix ordering, $\{C_N^{(\infty)}\}$ may or may not be proper. If it is proper, then $(V, \{C_n^{(\infty)}\}, e)$ will be an operator system satisfying the conditions of Theorem \ref{Theorem: SIC Characterization} or Theorem \ref{Theorem: MUB Characterization}, depending on the initial operator system $V$ considered. The details for SIC-POVMs and mutually unbiased bases are similar, so we give a detailed description for SIC-POVMs and then summarize the relevant differences for mutually unbiased bases.

For the remainder of this section, fix $d \in \mathbb{N}$ and let $\lambda = \frac{1}{d+1}$. Let $(V, \{C_n\}, e)$ be a $d^2$-dimensional operator system spanned by positive elements $p_1, p_2, \dots, p_{d^2} \in V$. We call $(V, \{C_n\}, e)$ a \textbf{SIC-system} if it satisfies the following conditions:
\begin{itemize}
    \item $\sum_i p_i = de$
    \item $p_i (p_j - \lambda e) p_i = 0$ abstractly  whenever $i \neq j$ (see Subsection \ref{subsec: Abstract relations}).
\end{itemize}
In this situation, we call the spanning vectors $\{p_1, p_2, \dots, p_n\}$ a \textbf{SIC-basis}. A SIC-system will serve as the initial data for our construction. Given a SIC-system $(V, \{C_n\}, e)$, we define an inductive limit of matrix orderings as follows. First, extend the basis $\{p_1, \dots, p_{d^2}\}$ to an infinite sequence $\{p_i\}_{i=1}^{\infty}$ by setting $p_j = p_k$ whenever $j = k + Nd^2$ for some $N \in \mathbb{N}$ (so that $p_{d^2+1} = p_1, p_{d^2+2}=p_2$ and so on). For every $n \in \mathbb{N}$, let $D_n^{(0)} := C_n$. We inductively define matrix orderings $\{D_n^{(k)}\}$: for each $k=0,1,2,\dots$, if $k=2j$ is even, set $D_n^{(k+1)} = D_n^{(k)}(p_j)$, and if $k$ is odd, set $D_n^{(k+1)} = (D_n^{(k)})^{d-\text{min}}$. Then $\{D_n^{(0)}\}_n, \{D_n^{(1)}\}_n, \{D_n^{(2)}\}_n, \dots$ is a nested sequence of matrix orderings. Define $\widetilde{C}_n := D_n^{(\infty)}$, where $\{D_n^{(\infty)}\}$ is the inductive limit of the nested sequence $\{D_n^{(k)}\}$.

\begin{proposition} \label{prop: Tilde C satisfies SIC Thm}
    Suppose that $(V, \{C_n\}, e)$ is a SIC-system, and that $\{\widetilde{C}_n\}$ is a proper matrix ordering. Then $(V, \{\widetilde{C}_n\}, e)$ satisfies the conditions of Theorem \ref{Theorem: SIC Characterization}.
\end{proposition}

\begin{proof}
    Assume $\{\widetilde{C}_n\}$ is proper. Since $\{D_n^{(0)}\}, \{D_n^{(1)}\}, \dots$ is nested, the same is true of any subsequence of matrix orderings, and the inductive limit of any subsequence is again $\{\widetilde{C}_n\}$. Thus $\{\widetilde{C}_n\}$ is also equal to the inductive limit of the nested sequence $\{D_n^{(2)}\}, \{D_n^{(4)}\}, \dots$ and since each of these matrix orderings is $d$-minimal, its inductive limit $\{\widetilde{C}_n\}$ is $d$-minimal by Theorem \ref{thm: inductive limits respect k-minimality}. Similarly, we see that each $p_i$ is an abstract projection in $(V, \{\widetilde{C}_n\}, e)$ by Theorem \ref{thm: inductive limits respect abstract projections} and Theorem \ref{thm: projection characterizations}, since $\{\widetilde{C}_n\}$ is the inductive limit of the sequence $\{D_n^{(2i+1)}\}, \{D_n^{(2i+2d^2+1)}\}, \{D_n^{(2i+4d^2+1)}\}, \dots$ and since $D_n^{(2i+2Nd^2+1)} = D_n^{(2i+2Nd^2)}(p_i)$ for every $N \in \mathbb{N}$. So the conditions of Theorem \ref{Theorem: SIC Characterization} are satisfied.
\end{proof}

\begin{remark}
    \emph{Leveraging Theorem \ref{thm: d-minimal characterization}, the sequence of matrix orderings in the construction of $\{\widetilde{C}_n\}$ can be re-written entirely in terms of their $d\textsuperscript{th}$ cone, i.e. we can consider instead the sequence $D_d^{(k)}$ and its limit $\widetilde{C}_d$. This approach is used in a similar construction described in Section 5 of \cite{AR2023Preprint}. Thus, our construction is inherently ``finite-dimensional,'' since all the required data is found in the sequence of cones $D_d^{(k)} \subseteq M_d(S)$ and $\dim(M_d(S)) < \infty$. For brevity, we omit the details of this approach here.}
\end{remark}

We now wish to show that when a SIC-POVM exists in dimension $d$, then the construction described above is ``successful,'' meaning that, for a reasonable choice of initial SIC-system $(V,\{C_n\},e)$, the cone $\{\widetilde{C}_n\}$ will be proper. We first record some helpful results.

\begin{lemma} \label{lem: extend projection maps}
    Suppose that $(V, \{C_n\}, e)$ and $(W, \{E_n\}, f)$ are operator systems, $q \in W$ is an abstract projection, and $\pi: V \to W$ is a ucp map such that $\pi(p)=q$. Then $\pi$ is also completely positive respect to the matrix ordering $\{C_n(p)\}$ on $V$.
\end{lemma}

\begin{proof}
    Let $x \in C_n(p)$. Then for every $\epsilon > 0$ there exists $t > 0$ such that 
    \[ y := \begin{pmatrix} x & x \\ x & x \end{pmatrix} + \epsilon (p \oplus p^{\perp}) \otimes I_n + t (p^{\perp} \oplus p) \otimes I_n \in C_{2n} \]
    by the definition of $C_n(p)$. Since $\pi$ is completely positive, 
    \[ \pi^{(2n)}(y) = \begin{pmatrix} \pi^{(n)}(x) & \pi^{(n)}(x) \\ \pi^{(n)}(x) & \pi^{(n)}(x) \end{pmatrix}  + \epsilon I_n \otimes (q \oplus q^{\perp}) + t I_n \otimes (q^{\perp} \oplus q) \in E_{2n}. \]
    It follows that $\pi^{(n)}(x) \in E_n(q)$. But $q$ is a projection, so $E_n(q)=E_n$. Hence $\pi^{(n)}(x) \in E_n$. Since $x \in C_n(p)$ was arbitrary, $\pi$ is completely positive with respect to $\{C_n(p)\}$.
\end{proof}

\begin{lemma} \label{lem: extend d-minimal maps}
    Suppose that $(V, \{C_n\}, e)$ and $(W, \{E_n\}, f)$ are operator systems, $\pi:V \to W$ is ucp and $W$ is $d$-minimal. Then $\pi$ is ucp with respect to the matrix ordering $\{(C_n)^{d-\text{min}}\}$ on $V$.
\end{lemma}

\begin{proof}
    By Theorem 3.7 of \cite{xhabli2012super}, we see that $\pi: V \to W$ is is completely positive if and only if it is $d$-positive. Since $\pi^{(k)}(C_k) \subseteq E_k$ for every $k \leq d$ and since $C_k = (C_k)^{d-\text{min}}$ for every $k \leq d$, $\pi$ is $d$-positive on $\{(C_n)^{d-\text{min}}\}$ and hence $\pi$ is completely positive.
\end{proof}

\begin{proposition} \label{prop: extend to tilde C}
    Suppose that $(V, \{C_n\}, e)$ is a SIC-system with SIC basis $\{p_1, p_2, \dots, p_{d^2}\}$. Let $(W, \{E_n\}, f)$ be another SIC-system with SIC-basis $\{q_1, q_2, \dots, q_{d^2}\}$ which satisfies the conditions of Theorem \ref{Theorem: SIC Characterization}. If the linear map $\pi: V \to W$ defined by $\pi(p_i) = q_i$ for every $i=1,2,\dots,d^2$ is completely positive with respect to $\{C_n\}$, then it is also completely positive with respect to $\{\widetilde{C}_n\}$.
\end{proposition}

\begin{proof}
    We proceed by induction on the sequence of matrix orderings $\{D_n^{(k)}\}$. Let $k \in \mathbb{N} \cup \{0\}$ and suppose that $\pi: V \to W$ is completely positive with respect to the matrix ordering $\{D_n^{(k)}\}$ on $V$. If $k$ is even, say $k=2(i+Nd^2)$, then $D_n^{(k+1)} = D_n^{(k)}(p_i)$. Since $\pi(p_i)=q_i$ and $q_i$ is an abstract projection in $W$, $\pi$ is completely positive with respect to $\{D_n^{(k+1)}\}$ by Lemma \ref{lem: extend projection maps}. If $k$ is odd, then $D_n^{(k+1)} = (D_n^{(k)})^{d-\text{min}}$. Since $W$ is $d$-minimal, then $\pi$ is completely positive with respect to $\{D_n^{(k+1)}\}$ by Lemma \ref{lem: extend d-minimal maps}. It follows that $\pi$ is completely positive with respect to every matrix ordering $\{D_n^{(k)}\}$. Now suppose that $x \in \widetilde{C}_n$ and let $\epsilon > 0$. Then there exists $k$ such that $x + \epsilon I_n \otimes e \in D_n^{(k)}$. Hence $\pi^{(n)}(x) + \epsilon I_n \otimes f \geq 0$. This holds for every $\epsilon > 0$ and therefore $\pi$ is completely positive with respect to $\{\widetilde{C}_n\}$. 
\end{proof}

We now consider the existence of a SIC-system $(V, \{C_n\}, e)$ such that the induced matrix ordering $\{\widetilde{C}_n\}$ is proper. Of course, if $(V, \{C_n\}, e)$ already satisfies the conditions of Theorem \ref{Theorem: SIC Characterization}, then $\widetilde{C}_n = C_n$ and there is nothing to show. However, we will see that we can explicitly construct a family of SIC-systems $(V, \{C_n\}, e)$ with the property that a SIC-POVM exists in dimension $d$ if and only if some member of the that family $(V, \{C_n\}, e)$ has a proper induced cone $\{\widetilde{C}_n\}$.

The family we will describe is indexed by the set of all increasing sequence $\vec{t} = \{t_i\}_{i=1}^\infty$ of positive real numbers, i.e. for every sequence $\vec{t} = \{t_i\}$ with $0 < t_1 < t_2 < \dots$ we will describe a corresponding SIC-system which we denote $(V, \{C_n(\vec{t})\}, e)$. First, define the $*$-vector space $V$ to be the diagonal $d^2 \times d^2$ matrices equipped with its usual adjoint. Define $p_i$ to be the matrix with a $1$ in the $i^\textsuperscript{th}$ diagonal entry and zeros elsewhere, i.e. the $i^\textsuperscript{th}$ diagonal matrix unit. Define the unit of $V$ to be the matrix $e = d^{-1}(\sum_i p_i)$, i.e. $d^{-1}$ times the identity matrix of $M_{d^2}$. 

We will define a matrix ordering $\{C_n(\vec{t})\}$ on $V$ as follows. We begin with the cone $C_1(\vec{t})$. For each pair $i,j \in \{1,2, \dots, d^2\}$ with $i \neq j$ and each $n \in \mathbb{N}$, define
\[ x_{i,j,n}^+ := (p_i - \lambda e) + \frac{1}{n} p_j + t_n p_j^{\perp} \quad \text{and} \quad x_{i,j,n}^- := (\lambda e - p_i) + \frac{1}{n} p_j + t_n p_j^{\perp}. \]
We let $C_1(\vec{t})$ denote the Archimedean closure of the cone generated by the elements $\{p_k, p_k^{\perp}, x_{i,j,n}^+, x_{i,j,n}^-\}$, i.e. $y \in C_1(\vec{t})$ if and only if for every $\epsilon > 0$
\[ y + \epsilon e = \sum_{k=1}^{d^2} (\alpha_k p_k + \beta_k) p_k^{\perp} + \sum_{n=1}^\infty \sum_{i\neq j} (\gamma_{i,j,n}^+ x_{i,j,n}^+ + \gamma_{i,j,n}^- x_{i,j,n}^-) \]
for some non-negative coefficients $\alpha_k, \beta_k, \gamma_{i,j,n}^+$, and $\gamma_{i,j,n}^-$ with only finitely many of these non-zero.

To extend $C_1(\vec{t})$ to a matrix ordering $\{C_n(\vec{t})\}$, we recall the \text{OMAX} operator system structure from \cite{PaulsenTodorovTomforde2011OSS}. Given a $*$-vector space $V$, a cone $C \subseteq V_h$ and an element $e \in V_h$, we define $(C)_n^{\text{max}}$ to be the set of all $x \in M_n(V)_h$ such that for every $\epsilon > 0$ there exist $Q_1, Q_2, \dots, Q_N \in M_n^+$ and $x_1, x_2, \dots, x_n \in C$ such that
\[ x + \epsilon I_n \otimes e = \sum_{i=1}^N Q_i \otimes x_i. \]
Paulsen, Tomforde, and Todorov prove the following concerning $\{(C)_n^{\text{max}}\}:$

\begin{lemma}[See Theorem 3.22 of \cite{PaulsenTodorovTomforde2011OSS}] \label{lem: OMAX properties}
The sequence $\{(C)_n^{\text{max}}\}$ enjoys the following properties:
\begin{enumerate}
    \item If $C$ is proper and Archimedean closed, then $(V, \{(C)_n^{\text{max}} \}, e)$ is an operator system satisfying $(C)_1^{\text{max}} = C$.
    \item If, in addition, $(W, \{E_n\}, f)$ is an operator system and $\pi: V \to W$ is a linear map such that $\pi(C) \subseteq E_1$, then $\pi$ is completely positive with respect to the matrix ordering $\{(C)_n^{\text{max}}\}$ on $V$.
\end{enumerate}
\end{lemma}

For each $n \in \mathbb{N}$, We define $C_n(\vec{t}) = (C_1(\vec{t}))_n^{\text{max}}$. By Lemma \ref{lem: OMAX properties}, whenever $C_1(\vec{t})$ is a proper Archimedean closed cone, $\{C_n(\vec{t})\}$ is a matrix ordering. We now wish to show that we can always find a sequence $\vec{t} = \{t_i\}$ such that $C_1(\vec{t})$ is proper and Archimedean closed. The proof makes use of a natural inner product structure on $V$. We start with an easy technical lemma.

\begin{lemma} \label{lem: positive inner product}
    Suppose that $V$ is a finite-dimensional real Hilbert space, $C \subseteq V$ is a cone such that $\langle x,y \rangle \geq 0$ for every $x,y \in C$, and $e \in C$ is a unit vector. Then the set $\widehat{C}$ consisting of all $y \in V$ such that $y + \epsilon e \in C$ for every $\epsilon > 0$ is a proper cone.
\end{lemma}

\begin{proof}
    It is clear that $\widehat{C}$ is a cone, so we check that it is proper. Suppose that $\pm y \in \widehat{C}$. Then for every $\epsilon > 0$ there exist $x_1, x_2 \in C$ such that $y + \epsilon e = x_1$ and $-y + \epsilon e = x_2$. Then
    \begin{eqnarray}
        0 & \leq & \langle x_1, x_2 \rangle \nonumber \\
        & = & -\|y\|^2 + \epsilon \langle y,e \rangle - \epsilon \langle e,y \rangle + \epsilon^2 \nonumber \\
        & = & \epsilon^2 - \|y\|^2. \nonumber
    \end{eqnarray}
    Hence $\|y\|^2 \leq \epsilon^2$ for every $\epsilon > 0$, so $y=0$.
\end{proof}

Consider the normalized Hilbert-Schmidt inner product on $M_d$ defined by $\langle A,B \rangle = \frac{1}{d} \Tr(A^*B)$. Recall that if $\{P_i\} \subseteq M_d$ is a SIC-POVM, then $\langle P_i, P_j \rangle = \lambda/d$ whenever $i \neq j$ and that $\|P_i\|_{HS}^2 = 1/d$ (where $\|A\|_{HS} := \langle A,A \rangle^{1/2}$). Finally note that the set of Hermitian matrices in $M_d$ is a real Hilbert space with respect to the Hilbert-Schmidt inner product.

We will prove the following theorem by ``artificially'' defining the Hilbert-Schmidt inner product on a SIC-system.

\begin{proposition} \label{prop: SIC cone}
    There exists a sequence $\vec{t} = \{t_i\}$ such that $0 < t_1 < t_2 < \dots$ for which $C_1(\vec{t})$ is proper and Archimedean closed.
\end{proposition}

\begin{proof}
    For each pair $p_i, p_j \in \{p_1, p_2, \dots, p_{d^2}\}$, define
    \[ \langle p_i, p_j \rangle = \begin{cases} 1/d & i = j \\ \lambda/d & i \neq j \end{cases} \]
    and extend (by sesquilinearity) to an inner product on $V$. This inner product is well-defined since $\{p_1, p_2, \dots, p_{d^2}\}$ is linearly independent in $V$. Moreover, the inner product is positive definite (i.e. $\langle x,x \rangle = 0$ implies $x=0$) since the matrix $(\langle p_i, p_j \rangle)_{i,j}$ has rank $d^2$ (see e.g. page 410 of \cite{Horn_Johnson_1985}). Finally, note that the restriction of this inner product to $V_h$ is real-valued and hence $V_h$ is a real Hilbert space.

    Recall that, given a sequence $\vec{t}=\{t_n\}$, we define \[ x_{i,j,n}^+ := (p_i - \lambda e) + \frac{1}{n} p_j + t_n p_j^{\perp} \quad \text{and} \quad x_{i,j,n}^- := (\lambda e - p_i) + \frac{1}{n} p_j + t_n p_j^{\perp}. \] Define a cone $C \subseteq V_h$ to be the cone generated by the elements $\{p_k, p_k^{\perp}, x_{i,j,n}^+, x_{i,j,n}^-\}$ where $i,j$ and $k$ range over $\{1,2,\dots, d^2\}$ with $i \neq j$ and $n \in \mathbb{N}$, i.e.
    \[ C = \left\{ \sum_{k=1}^{d^2} (\alpha_k p_k + \beta_k p_k^{\perp}) + \sum_{n=1}^\infty \sum_{i \neq j} (\gamma_{i,j,n}^+ x_{i,j,n}^+ + \gamma_{i,j,n}^- x_{i,j,n}^-) \right\} \]
    where $\{\alpha_k, \beta_k, \gamma_{i,j,n}^+, \gamma_{i,j,n}^-\}$ are non-negative coefficients, only finitely many of which are non-zero.
    
    We claim that the sequence $\vec{t}$ can be chosen such that $\langle x,y \rangle \geq 0$ for any $x,y \in C$. It suffices to show that a sequence $\{t_i\}$ can be chosen such that for any $x,y$ in the generating set $\{p_k, p_k^{\perp}, x_{i,j,n}^+, x_{i,j,n}^-\}$ we have $\langle x,y \rangle \geq 0$. Using the definition of the inner product and that $e = (\sum p_i) / d$ and $p_i^{\perp} = e - p_i$, we have that
    \[ \langle p_i, p_j \rangle = \begin{cases} 1/d & i = j \\ \lambda/d & i \neq j \end{cases}, \quad \langle p_i, p_j^{\perp} \rangle = \begin{cases} 0 & i = j \\ (1-\lambda)/d & i \neq j \end{cases}, \quad \langle p_i^{\perp}, p_j^{\perp} \rangle = \begin{cases} (d-1)/d & i = j \\ (d-2 + \lambda)/d & i \neq j \end{cases} \]
    are valid. For the remaining inner products, first note that $\langle p_i, \lambda e - p_j \rangle = 0$ whenever $i \neq j$. Hence
    \[ \langle p_k, x_{i,k,n}^+ \rangle = \langle p_k, x_{i,k,n}^- \rangle = \frac{1}{nd}. \]
    Additionally,
    \[ \|\lambda e - p_i\|^2 = \langle \lambda e - p_i, \lambda e - p_i \rangle = \dfrac{\lambda^2d - 2\lambda + 1}{d} \]
    and thus
    \[ \| x_{i,j,n}^* - t_n p_j^{\perp} \|^2 = \| \pm (\lambda e - p_i) + \frac{1}{n} p_j \|^2 = \frac{\lambda^2 d - 2\lambda +1}{d} + \frac{1}{nd} \leq \frac{\lambda^2d - 2\lambda + 2}{d} \]
    where $x_{i,j,n}^* = x_{i,j,n}^+$ or $x_{i,j,n}^-$. Let $\beta = ((\lambda^2 d - 2\lambda + 2)/d)^{1/2}$. Then for $j \neq k$,
    \begin{eqnarray}  \langle p_k, x_{i,j,n}^* \rangle & = & \langle p_k, t_n p_j^{\perp} + (x_{i,j,n}^* - t_n p_j^{\perp}) \rangle \nonumber \\
    & \geq & t_n \frac{1-\lambda}{d} - \frac{\beta}{\sqrt{d}} \nonumber
    \end{eqnarray}
    by Cauchy-Schwarz, since $\| x_{i,j,n}^* - t_n p_j^{\perp} \| \leq \beta$. This is positive provided that $t_n > \frac{\sqrt{d}\beta}{1-\lambda}$ for every $n$. Similarly,
    \begin{eqnarray}  \langle p_k^{\perp}, x_{i,j,n}^* \rangle & = & \langle p_k^{\perp}, t_n p_j^{\perp} + (x_{i,j,n}^* - t_n p_j^{\perp}) \rangle \nonumber \\
    & \geq & t_n \frac{d-2+\lambda}{d} - \frac{\sqrt{d-1}\beta}{\sqrt{d}} \nonumber
    \end{eqnarray}
    which is positive provided $t_n \geq \frac{\sqrt{d^2-d}\beta}{d-2+\lambda}$. 

    Finally, we consider $\langle x_{i,j,n}^*, x_{a,b,m}^* \rangle$. We will use the estimate $\| x_{i,j,n}^* - t_n p_j^{\perp} \| \leq \beta$ to show that the inner product is positive. Let $\alpha = (d-2+\lambda)/d$, and let $\gamma := (d-1)^{1/2} \beta / d^{1/2}$. Then
    \begin{eqnarray} \langle x_{i,j,n}^*, x_{a,b,m}^* \rangle & = & \langle t_n p_j^{\perp} + (x_{i,j,n}^* - t_n p_j^{\perp}), t_m p_b^{\perp} + (x_{a,b,m}^* - t_m p_b^{\perp}) \rangle \nonumber \\
    & \geq & t_n t_m \alpha - (t_n + t_m) \gamma - \beta^2 \label{eqn: number 3}
    \end{eqnarray}
    using Cauchy-Schwarz and the estimate $\| x_{i,j,n}^* - t_n p_j^{\perp} \| \leq \beta$. We see that (\ref{eqn: number 3}) is positive if $n = m$ and \[ t_n \geq \frac{2\gamma + \sqrt{ 4\gamma^2 + 4\alpha \beta^2}}{2\alpha} > \gamma / \alpha \]
    using the quadratic formula. When $n < m$, we have $\epsilon := t_m - t_n > 0$ whenever $\{t_k\}$ is increasing. Hence
    \[ \langle x_{i,j,n}^*, x_{a,b,m}^* \rangle \geq t_n^2 \alpha - 2t_n \gamma - \beta^2 + \epsilon (\alpha t_n - \gamma) > 0 \]
    since $t_n > \gamma / \alpha$.

    We conclude that, by choosing suitably large values for $\{t_i\}$, we have $\langle x,y \rangle \geq 0$ for all $x,y \in C$. It follows that $C_1(\vec{t}) = \widehat{C}$ is a proper cone by Lemma \ref{lem: positive inner product}.
\end{proof}

\begin{theorem}
    There exists an increasing sequence $\vec{t}$ such that $\{C_n(\vec{t})\}$ is a proper matrix ordering and hence $(V, \{C_n(\vec{t})\}, e)$ is a SIC-system. Moreover, if there exists a SIC-POVM in dimension $d$, then there exists a sequence $\vec{t}$ such that $\{\widetilde{C}_n(\vec{t})\}$ is proper and hence $(V, \{\widetilde{C}_n(\vec{t})\}, e)$ satisfies the conditions of Theorem \ref{Theorem: SIC Characterization}.
\end{theorem}

\begin{proof}
    By Proposition \ref{prop: SIC cone}, we can find a sequence $\vec{t}$ such that $\{C_n(\vec{t})\}$ is proper.

    Now suppose there exists a SIC-POVM in dimension $d$. Let $\{P_i\}_{i=1}^{d^2} \subseteq M_d$ denote the corresponding projections. Then for every $n \in \mathbb{N}$, there exists $t_n > 0$ such that 
    \[ \pm (P_i - \lambda I) + \frac{1}{n} P_j + t_n P_j^{\perp} \geq 0 \]
    since $P_j (P_i - \lambda I) P_j = 0$ and since $P_j$ is a projection. Letting $\vec{t} = \{t_k\}$ be a sequence satisfying this property, we see that the map $\pi: p_i \mapsto P_i$ from $V$ to $M_d$ is positive on the cone $C$ generated by elements $\{p_k, p_k^{\perp}, x_{i,j,n}^+, x_{i,j,n}^-\}$. Consequently $\pi(C_1(\vec{t})) \subseteq M_d^+$. By Lemma \ref{lem: OMAX properties}, $\pi$ is completely positive on the SIC-system $(V, \{C_n(\vec{t})\}, e)$. Since $M_d$ is $d$-minimal and since each $P_i$ is a projection, it follows that $\pi$ is completely positive on the matrix ordering $\{ \widetilde{C}_n(\vec{t})\}$ by Proposition \ref{prop: extend to tilde C}. Since $\pi$ is one-to-one and $\pi( \widetilde{C}(\vec{t})_1) \subseteq M_d^+$ and since $M_d^+$ is a proper cone, $\{ \widetilde{C}_n(\vec{t})\}$ is a proper matrix ordering. Hence $(V, \{ \widetilde{C}_n(\vec{t})\}, e)$ satisfies the conditions of Theorem \ref{Theorem: SIC Characterization}, by Proposition \ref{prop: Tilde C satisfies SIC Thm}.
\end{proof}

The above Theorem implies that the existence question for SIC-POVMs is equivalent to asking if there exists a sequence $\vec{t}$ such that the matrix ordering $\{ \widetilde{C}_n(\vec{t})\}$ is proper. While the initial cone $C_1(\vec{t})$ is easily described, we don't understand enough about the induced matrix ordering $\{ \widetilde{C}_n(\vec{t})\}$ to say which sequences $\vec{t}$ will induce a proper matrix ordering. We hope to study this question in more detail in future work.

We conclude by discussing the existence question for $d+1$ MUBs in dimension $d$. A construction similar to the one discussed above for SIC-systems can be carried out for MUBs. We will not go through the details since it is similar to the case for SIC-POVMs. Instead we will describe the basic definitions and leave it to the interested reader to check that the corresponding constructions hold.

Let $d \in \mathbb{N}$ and let $\mu = \frac{1}{d}$. We call an operator system $W$ with unit $e$ spanned by vectors $\{p_{i}^x : i=1,2,\dots,d; x=1,2,\dots,d+1\}$ a \textbf{MUB-system} if it satisfies the conditions
\begin{itemize}
    \item $\sum_i p_i^x = e$ for each $x=1,2,\dots,d+1$, and
    \item $p_i^x (p_j^y - \mu e)p_i^x = 0$ abstractly whenever $x \neq y$.
\end{itemize}
The condition $\sum_i p_i^x = e$ implies that $\dim(W) \leq d^2$ since $W$ is spanned by $\{e\} \cup \{p_i^x : i=1,2,\dots,d-1; x=1,2,\dots,d+1\}$. We can construct examples of dimension $d^2$ in the following way. As a $*$-vector space, define $W$ to be the $d^2 \times d^2$ diagonal matrices. Let $D_k$ denote the $k \times k$ diagonal matrices and $E_j$ denote the matrix unit with a 1 in its $j\textsuperscript{th}$ diagonal entry and zeroes elsewhere. Then we have $D_{d^2} \cong (D_{d-1} \otimes D_{d+1}) \oplus D_1$ via the Kronecker product. We define $p_i^x := (E_i \otimes E_x) \oplus 0$ for each $i \leq d-1$ and $x \leq d+1$, and we define $p_d^x = I_{d^2} - (I_{d-1} \otimes E_x) \oplus 0$. This defines a $d^2$-dimensional$*$-vector space with unit $e = I_{d^2}$ and self-adjoint generators $\{p_i^x\}$ satisfying the condition $\sum_i p_i^x = e$. The positive cone can be defined by setting $C(\vec{t})$ equal to the closure of the cone generated by elements $\{ p_i^x, y_{x,i,y,j,n}^+, y_{x,i,y,j,n}^-\}$ where 
\[ y_{x,i,y,j,n}^+ := (p_i^x - \mu e) + \frac{1}{n} p_j^y + t_n (p_j^y)^{\perp} \quad \text{and} \quad y_{x,i,y,j,n}^- := (\mu e - p_i^x) + \frac{1}{n} p_j^y + t_n (p_j^y)^{\perp}. \]
As with SIC-POVMs, a sequence $\vec{t}$ can be chosen so that the resulting cone is proper. Using the $\text{OMAX}$ construction (see Lemma \ref{lem: OMAX properties}), one can then generate a MUB-system $(W,\{C_n(\vec{t})\},e)$. Moreover, one can show that $d+1$ mutually unbiased bases exist in dimension $d$ if and only if $(W,\{\widetilde{C}_n(\vec{t})\},e)$ satisfies the conditions of Theorem \ref{Theorem: MUB Characterization} for some sequence $\vec{t}$. Here the matrix ordering $\{\widetilde{C}_n(\vec{t})\}$ is constructed from $\{C_n(\vec{t})\}$ in exactly the same way as for SIC-systems (see the text preceding Lemma \ref{lem: OMAX properties} above). We leave it to the interested reader to investigate the details.

\section*{Declarations}

\subsection*{Funding}

This work was supported in part by a grant from the Texas Christian University Research and Creative Activities Fund. 

\subsection*{Competing interests}

The author has no financial or proprietary interests in any material discussed in this article.

\bibliographystyle{plain}
\bibliography{Refs}

\begin{thebibliography}{10}

\bibitem{AR2020Published}
Roy Araiza and Travis Russell.
\newblock An abstract characterization for projections in operator systems.
\newblock {\em Journal of Operator Theory}, 90(1):41--72, 2023.

\bibitem{AR2023Preprint}
Roy Araiza and Travis Russell.
\newblock Operator systems generated by projections.
\newblock {\em arXiv preprint arXiv:2001.04383}, 2023.

\bibitem{ARTaPublished2022}
Roy Araiza, Travis Russell, and Mark Tomforde.
\newblock A universal representation for quantum commuting correlations.
\newblock {\em Annales Henri Poincar{\'e}}, 23:4489--4520, 2022.

\bibitem{ART2024Published}
Roy Araiza, Travis Russell, and Mark Tomforde.
\newblock Matricial archimedean order unit spaces and quantum correlations.
\newblock {\em Indiana University Mathematics Journal}, 72(6):2567--2591, 2023.

\bibitem{BenedettoFickusFrames2003}
J.J. Benedetto and M.~Fickus.
\newblock Finite normalized tight frames.
\newblock {\em Advances in Computational Mathematics}, 18:357--385, 2003.

\bibitem{CHOIEffros1977}
Man-Duen Choi and Edward~G. Effros.
\newblock Injectivity and operator spaces.
\newblock {\em Journal of Functional Analysis}, 24(2):156--209, 1977.

\bibitem{DavidsonCstarBook}
Kenneth~R. Davidson.
\newblock {\em {$C^*$}-algebras by example}, volume~6 of {\em Fields Institute Monographs}.
\newblock American Mathematical Society, Providence, RI, 1996.

\bibitem{PaulsenGroups}
Douglas Farenick, Ali~S. Kavruk, Vern~I. Paulsen, and Ivan~G. Todorov.
\newblock Operator systems from discrete groups.
\newblock {\em Communications in Mathematical Physics}, 281:207--238, 2014.

\bibitem{FarkasEtAl2024Preprint}
Mate Farkas, Jurij Volcic, Sigurd Storgaard, Ranyiliu Chen, and Laura Mancinska.
\newblock Maximal device-independent randomness in every dimension.
\newblock {\em arXiv:2409.18916}, 2024.

\bibitem{Gribling2024mutuallyunbiased}
Sander Gribling and Sven Polak.
\newblock Mutually unbiased bases: polynomial optimization and symmetry.
\newblock {\em {Quantum}}, 8:1318, April 2024.

\bibitem{hamana1979injective}
Masamichi Hamana.
\newblock Injective envelopes of operator systems.
\newblock {\em Publications of the Research Institute for Mathematical Sciences}, 15(3):773--785, 1979.

\bibitem{Horn_Johnson_1985}
Roger~A. Horn and Charles~R. Johnson.
\newblock {\em Matrix Analysis}.
\newblock Cambridge University Press, 1985.

\bibitem{HorodeckiRZOpenProbsQIT}
Pawe\l{} Horodecki, \L{}ukasz Rudnicki, and Karol \ifmmode~\dot{Z}\else \.{Z}\fi{}yczkowski.
\newblock Five open problems in quantum information theory.
\newblock {\em PRX Quantum}, 3:010101, Mar 2022.

\bibitem{ji2020mip}
Zhengfeng Ji, Anand Natarajan, Thomas Vidick, John Wright, and Henry Yuen.
\newblock {MIP*}= {RE}.
\newblock {\em arXiv preprint arXiv:2001.04383}, 2020.

\bibitem{Navascues2012Handbook}
Miguel Navascu{\'e}s, Stefano Pironio, and Antonio Ac{\'i}n.
\newblock {\em SDP Relaxations for Non-Commutative Polynomial Optimization}, pages 601--634.
\newblock Springer US, Boston, MA, 2012.

\bibitem{paulsen2002completely}
Vern~I Paulsen.
\newblock {\em Completely bounded maps and operator algebras}, volume~78.
\newblock Cambridge University Press, 2002.

\bibitem{PaulsenTodorovTomforde2011OSS}
Vern~I. Paulsen, Ivan~G. Todorov, and Mark Tomforde.
\newblock Operator system structures on ordered spaces.
\newblock {\em Proceedings of the London Mathematical Society}, 102(1):25--49, 2011.

\bibitem{RaynalLuEnglertMUBsSix2011}
Philippe Raynal, Xin L\"u, and Berthold-Georg Englert.
\newblock Mutually unbiased bases in six dimensions: The four most distant bases.
\newblock {\em Phys. Rev. A}, 83:062303, Jun 2011.

\bibitem{SICs2004}
Joseph~M. Renes, Robin Blume-Kohout, A.~J. Scott, and Carlton~M. Caves.
\newblock {Symmetric informationally complete quantum measurements}.
\newblock {\em Journal of Mathematical Physics}, 45(6):2171--2180, 06 2004.

\bibitem{WOOTTERS1989363}
William~K Wootters and Brian~D Fields.
\newblock Optimal state-determination by mutually unbiased measurements.
\newblock {\em Annals of Physics}, 191(2):363--381, 1989.

\bibitem{xhabli2012super}
Blerina Xhabli.
\newblock The super operator system structures and their applications in quantum entanglement theory.
\newblock {\em Journal of Functional Analysis}, 262(4):1466--1497, 2012.

\bibitem{ZaunerThesis}
G.~Zauner.
\newblock Grundz\"uge einer nichtkommutativen designtheorie, 1999.

\end{thebibliography}

\end{document}